\documentclass[11pt,letterpaper]{article}

\usepackage{amssymb}
\usepackage{amsthm}
\usepackage{amsmath}            
\usepackage{amscd}

\newcommand{\tr}{\text{tr} \,}

\newcommand{\R}{{\mathbb R}}

\newcommand{\ra}{{\: \rightarrow \:}}
\newcommand{\Div}{{\text{div}}}

\newtheorem{thm}{Theorem}[section]
\newtheorem{prop}[thm]{Proposition}
\newtheorem{lem}[thm]{Lemma}

\newtheorem{defn}[thm]{Definition}

\begin{document}



\title{Maximal Hypersurfaces in Spacetimes with Translational Symmetry}
\author{Andrew Bulawa}
\maketitle

\begin{abstract}
We consider four-dimensional vacuum spacetimes which admit a nonvanishing spacelike Killing field.  The quotient with respect to the Killing action is a three-dimensional quotient spacetime $(M,g)$.  We establish several results regarding maximal hypersurfaces (spacelike hypersurfaces of zero mean curvature) in such quotient spacetimes.  First, we show that a complete noncompact maximal hypersurface must either be a cylinder $S^1 \times \R$ with flat metric or else conformal to the Euclidean plane $\R^2$.  Second, we establish positivity of mass for certain maximal hypersurfaces, referring to a analogue of ADM mass adapted for the quotient setting.  Finally, while lapse functions corresponding to the maximal hypersurface gauge are necessarily bounded in the four-dimensional asymptotically Euclidean setting, we show that nontrivial quotient spacetimes admit the maximal hypersurface gauge only with unbounded lapse.
\end{abstract}

\section{Introduction}

Consider a four-dimensional manifold $\hat M$ with smooth Lorentzian metric $\hat g$ of signature $(-,+,+,+)$.  Suppose $(\hat M,\hat g)$ admits a smooth nonvanishing spacelike Killing field $X$, so that the quotient with respect to the Killing action is a smooth three-dimensional Lorentzian manifold $(M,g)$.  It is well known (see \cite{rG71}) that the vacuum Einstein equations $\hat Ric = 0$ for $\hat g$ can be written in the following equivalent form with respect to the quotient metric $g$:
\begin{equation}
\label{reve}
\left.
\begin{aligned}
	Ric &= D d \psi + d\psi \otimes d\psi + \frac{1}{2} e^{-4\psi}\left( \omega \otimes \omega - |\omega|_g^2 \cdot g \right) \qquad \\
	\Box_g \psi + |d\psi|_g^2 &= -\frac{1}{2} e^{-4\psi} |\omega|_g^2 \\
	\Div (\omega) &= 3 \langle \omega, d \psi \rangle
\end{aligned}\right\rbrace
\end{equation}
$e^\psi$ denotes the length of $X$ and $\omega$ denotes the twist one-form characterizing the integrability of the horizontal distribution of the bundle $\hat M \ra M$ (see Section \ref{quotientsec} for details).  $D$ denotes the covariant derivative given by $g$, $\Box_g \psi$ denotes the $g$-trace of the Hessian $Dd\psi$, and $\langle \, , \, \rangle$ denotes the pointwise defined inner product $g$ induces on one-forms.  We call any smooth three-dimensional Lorentzian manifold $(M,g)$ satisfying \eqref{reve} a \emph{quotient vacuum spacetime}.

The focus of this work is to better understand properties of maximal hypersurfaces in quotient vacuum spacetimes.  This is motivated by the results of \cite{AV94}, which laid out a Hamiltonian framework and explored the concept of ADM mass in the quotient setting; also the follow-up paper \cite{mV95}, which explored foliations by maximal hypersurfaces.  The maximal hypersurface assumption is of well-known interest primarily because it decouples the Hamiltonian and momentum constraint equations in the initial value formulation of general relativity (see \cite{C-BY80}).

Recall that the ADM formalism in the four-dimensional setting gives rise to the ADM mass when one considers lapse functions which asymptotically approach 1 at spacelike infinity on an asymptotically Euclidean spacelike hypersurface (see \cite{rW84} for details).  Moreover, the maximal hypersurface gauge enforces the condition $\Delta u \sim 0$ for the lapse function $u$ near spacelike infinity in a vacuum spacetime (or one satifsying an appropriate energy condition), so $u$ must be bounded on each hypersurface.  The situation is different in quotient spacetimes.  In that situation, lapse functions which are unbounded give rise to the analogue of ADM mass, and, in fact, lapse functions which exhibit $\log$-like growth are a required feature of foliations by asymptotically flat \emph{maximal} hypersurfaces \cite{AV94}.  These disparities between the four- and three-dimensional settings encourage us to explore in more detail foliations of quotient vacuum spacetimes by maximal hypersurfaces and the behavior of their lapse functions.  We do so without imposing an asymptotic flatness assumption.

The first result of this work addresses topological and geometric restrictions the maximality condition imposes on a single hypersurface, with no foliation necessarily in mind.  Let $\Sigma$ be a noncompact orientable spacelike hypersurface of $(M,g)$.  Assume $\Sigma$ is geodesically complete with respect to the hypersurface metric $h$ induced by $g$.  Note that maximality with respect to $\hat g$ means maximality of $\hat\Sigma = \pi^{-1}(\Sigma)$ as a hypersurface of $(\hat M,\hat g)$, where $\pi:\hat M \ra M$ denotes the quotient map corresponding to the Killing action.

\begin{thm}
\label{topgeo}
If $\Sigma$ is maximal with respect to $\hat g$, $g$, or the conformally rescaled metric $\tilde g = e^{2\psi}g$, then one of the following two statements holds:
\begin{itemize}
	\item $(\Sigma, h)$ is conformally equivalent to the Euclidean plane.  
	\item $(\Sigma, h)$ is a cylinder $S^1 \times \R$ with flat metric.
\end{itemize}
\end{thm}

In addition to the metric $g$, we consider the metrics $\hat g$ and $\tilde g$ because of the direct role the former plays in our framework and because of the significant simplification the latter grants to the system \eqref{reve}.  See, for example, equation \eqref{cric}.  Besides, it is not clear which metric is the most natural or useful once one moves to the quotient setting.

Positive mass theorems for quotient vacuum spacetimes were established in \cite{AV94} and \cite{BCM95}.  The next theorem, Theorem \ref{pmt}, extends these results to the situation in which the asymptotic flatness assumption is replaced by the maximal hypersurface assumption.  The fact that one can make this trade-off together with the conclusions of Theorem \ref{topgeo} illustrate that asymptotic flatness and maximality are intertwined to some extent.  This may not be surprising given the restrictions inherent in working in three spacetime dimensions, but we find it interesting to ask to what extent these conditions are related.  Proposition \ref{areaprop} touches on this point slightly, showing that the circumferences of large geodesic discs on maximal hypersurfaces cannot grow significantly greater than those in Euclidean space.

The ADM mass of an asymptotically flat hypersurface is adapted for quotient spacetimes in \cite{AV94}.  Instead of using Euclidean space as the model for asymptotic flatness for a hypersurface metric (two-dimensional in the quotient setting), the authors find a more appropriate model to be the flat metric induced on a cone in Euclidean space, given in standard polar coordinates $(r,\theta)$ as $r^{-\beta}(dr^2+r^2d\theta^2)$.  The constant $\beta$ is shown to be a constant multiple of the ADM mass, and this leads to a positive mass theorem: $\beta \geq 0$ with equality if and only if the hypersurface metric and the second fundamental form comprise initial data for Minkowski spacetime.  The authors find it necessary to require that $\beta \leq 2$ to ensure that the Hamiltonian framework remains well defined.  They point out that this upper bound contrasts the behavior of the ADM mass in four-dimensional spacetimes.

An application of the Gauss-Bonnet theorem shows that $\beta$ is directly related to the total scalar curvature of the hypersurface with respect to which it is defined.  In Section \ref{topsec} we use this fact to define $\beta$ for maximal hypersurfaces in a way which does not require asymptotic flatness but that reduces to the definition in \cite{AV94} when asymptotic flatness is assumed.  This definition is also analogous to the quantity $\theta_0$ in \cite{BCM95} where cylindrical symmetry is assumed.  Let $\hat h$ and $\hat k$ denote the metric and second fundamental form, respectively, that $\hat g$ induces on the hypersurface $\hat\Sigma = \pi^{-1}(\Sigma)$.  We prove the following:

\begin{thm}
\label{pmt}
Suppose that $\hat\Sigma$ is maximal with respect to $\hat g$.  Then 
\[
0 \leq \beta \leq 2,
\]
and $\beta = 0$ if and only if $\hat h$ is flat and $\hat k$ is identically zero.
\end{thm}

It does not appear to be as straightforward to establish this theorem when $\Sigma$ is instead maximal with respect to $g$ or $\tilde g$.  These cases are not addressed in this paper.

The final result of this paper concerns the maximal hypersurface gauge condition on quotient vacuum spacetimes.  That is, one assumes $M = \Sigma \times \R$ with each $\Sigma_t = \Sigma \times \{t\}$ a maximal spacelike hypersurface.  The lapse function is given by $u=(-|dt|^2)^{-1/2}$, where $t$ is projection $\Sigma \times \R \ra \R$ onto the $\R$-factor and the norm $|\cdot |$ is taken with respect to whichever metric, $\hat g$, $g$, or $\tilde g$, one has in mind.  Theorem \ref{maxgaugethm} illustrates that, in contrast to the four-dimensional asymptotically Euclidean setting, the maximal hypersurface gauge in nontrivial quotient spacetimes is not compatible with a bounded lapse.  This was observed in \cite{mV95} under the additional assumption of asymptotic flatness.

\begin{thm}
\label{maxgaugethm}
Suppose $(\hat M,\hat g)$ gives rise to a quotient vacuum spacetime $(\Sigma \times \R,g)$ on which $\omega$ vanishes identically.  Suppose each $\Sigma_t = \Sigma \times \{t\}$ is a maximal hypersurface that is noncompact, orientable, and complete with respect to the hypersurface metric induced by $g$.  If, on each $\Sigma_t$, the lapse is bounded above and $\psi$ is bounded above and below, then $\hat g$ is flat.
\end{thm}

Finally, we take a moment to mention the results of Huneau \cite{Hun16,Hun16stab} which solve the constraint equations for quotient vacuum spacetimes and establish the nonlinear stability of Minkowski spacetime in the quotient setting.  The latter result was achieved using generalized wave coordinates rather than the maximal hypersurface gauge explored here.  The argument to solve the constraint equations begins with the assumption that $(\Sigma,h)$ is conformal to the Euclidean plane.  The mean curvature $\tau$ of $(\Sigma,h)$ is allowed to vary (ruling out the maximal hypersurface gauge) in order to guarantee existence of solutions to the Lichnerowicz equation.  We point out, however, that while existence results appear not to be known for the case $\tau \equiv 0$, existence of solutions is not ruled out.

\section{The Quotient Structure}
\label{quotientsec}

Let $(M,g)$ be a quotient vacuum spacetime arising from a smooth vacuum spacetime $(\hat M, \hat g)$, as above.  Throughout the paper we assume $\Sigma \subset M$ is a smooth orientable spacelike hypersurface with induced Riemannian metric $h$ and corresponding connection $\nabla$.  On $\Sigma$ the tangent space $TM$ splits as $T\Sigma \oplus T\Sigma^\perp$, where $T\Sigma^\perp$ denotes the space of tangent vectors normal to $\Sigma$.  We have the second fundamental form $k: T\Sigma \ra T\Sigma$ given by
\[ k(Y) = - D_Y \nu, \]
where $\nu$ is a unit length section of $T\Sigma^\perp$.

$\Sigma$ is the quotient, with respect to the Killing action, of a spacelike hypersurface $\hat\Sigma \subset \hat M$ having metric $\hat h$ and second fundamental form $\hat k$.  The situation is summarized by the following diagram in which horizontal arrows denote embeddings and $\pi_M$ and $\pi_\Sigma$ denote the quotient maps corresponding to the Killing action:
\[
\begin{CD}
(\hat\Sigma,\hat h,\hat k) @>>> (\hat M,\hat g) \\
\pi_{\Sigma}@ VVV \pi_M@ VVV \\
(\Sigma, h, k) @>>> (M, g) \\
\end{CD}
\]
$\hat k$ is defined so that it equals the pullback $\pi_\Sigma^* k$.

\subsection*{$\psi$ and the twist one-form $\omega$}
\label{reduction}

Let $\xi$ denote the one-form dual to the Killing vector field $X$ on $\hat M$.  Define $\hat\omega = \ast \, (\xi \wedge d \xi)$, where $\,\ast\,$ denotes the Hodge star operator.  $\hat\omega$ vanishes identically on $\hat M$ if and only if $X$ is hypersurface orthogonal or, equivalently, the horizontal distribution of the total space $\hat M$ is integrable (see \cite{rW84}).  It is easily verified that $\hat\omega(X)$ and the Lie derivative $\mathcal L_X \hat\omega$ are identically zero.  Therefore $\hat\omega$ decends to a one-form $\omega$, called the \emph{twist one-form}, on $M$.  Quotient spacetimes on which $\omega$ vanishes identically are called \emph{polarized} in the literature.  We let $\psi$ denote the log of the length of the Killing vectorfield $X$:
\[
e^{2\psi} = \hat g(X,X)
\]
 
\subsection*{Relating the Geometry of $\hat\Sigma$ and $\Sigma$}
\label{relations}

It will be useful to establish relationships between quantities in the total space and quantities in the quotient.  The mean curvature $\tau$ of $\Sigma$ and the square norm of $k$ are given by
\[ \tau = \tr_h k \qquad \text{and} \qquad |k|_h^2 = \tr_h (k \circ k), \]
and analogous quantities for $\hat \Sigma$ are given by
\[ \hat \tau = \tr_{\hat h} \hat k \qquad \text{and} \qquad |\hat k|_{\hat h}^2 = \tr_{\hat h} (\hat k \circ \hat k). \]
where $\tr_h$ and $\tr_{\hat h}$ denote traces with respect to the indicated metrics.

$\hat \tau$ and $|\hat k|^2$ are invariant with respect to the Killing action, so they descend to functions on $\Sigma$, and we will view them as such.  In fact, any scalar on $\hat\Sigma$ determined by $\hat g$ or $\hat h$ is invariant with respect to the Killing action and will descend to a scalar on $\Sigma$.  We shall therefore view all scalars on $\hat\Sigma$ in this paper as scalars on $\Sigma$.

The following relationships are established in \cite{sD08}:
\begin{gather}
\label{ktokhat}	|\hat k|_{\hat h}^2 = |k|_h^2 + d\psi(\nu)^2 + \frac{1}{2}e^{-4\psi}\left( |\omega |_g^2 + \omega(\nu)^2 \right) \\
\label{tautotauhat} \hat \tau = \tau - d\psi(\nu)
\end{gather}

\subsection*{The Conformally Rescaled Geometry}
\label{confgeo}

We will also consider the conformally rescaled metrics $\tilde g$ on $M$ and $\tilde h$ on $\Sigma$ given by
\[
\tilde g = e^{2\psi} g \quad \text{and} \quad \tilde h = e^{2\psi} h.
\]
There are two reasons for considering this conformal rescaling.  First, the scalar curvature $\tilde s$ given by $\tilde h$ is always nonegative, placing the geometry of the surface $(\Sigma,\tilde h)$ in a significantly more manageable class.  Secondly, examples of quotient spacetimes such as Einstein-Rosen waves (see, for example, \cite{ABS971,kT65}) simplify considerably when viewed in this conformally rescaled setting.  We point out, however, that geodesic completeness of $\Sigma$ with respect to $\tilde h$ is not guaranteed.

The Ricci curvature $\tilde Ric$ corresponding to $\tilde g$ satisfies
\begin{equation}
	\label{cric}
	\tilde Ric = 2 d\psi \otimes d\psi + \frac{1}{2} e^{-4\psi} \omega \otimes \omega.
\end{equation}
This follows from \eqref{reve} and standard formulas regarding conformal rescalings.

The metric $\tilde g$ induces on $\Sigma$ the second fundamental form $\tilde k$ and the mean curvature
\[ \tilde \tau = \tr_{\tilde h} \tilde k, \]
where throughout the paper the convention is used that a tilde $\, \tilde{} \,$ or a subscript $\tilde g$ indicates that the marked quantity is defined with respect to the conformal metric $\tilde g$.  It is straightforward to derive the following relations:
\begin{gather}
\label{ktoktilde} \tilde k = e^\psi ( k - d\psi(\nu) \cdot Id_{T\Sigma} ) \\
\label{tautotautilde} \tilde \tau = e^{-\psi} ( \tau - 2 d\psi(\nu) ),
\end{gather}
where $Id_{T\Sigma} : T\Sigma \ra T\Sigma$ denotes the identity map.

\section{Topology and Geometry of Maximal Hypersurfaces}
\label{topsec}

In this section we prove Theorem \ref{topgeo}, but we begin with some preliminary steps.  A spacelike hypersurface is said to be \emph{maximal} if its mean curvature is identically zero.  We have discussed three mean curvatures, $\hat\tau$, $\tau$, and $\tilde\tau$, with respect to which the definition of maximality could refer.  We will say that $\Sigma$ is maximal if \emph{any one} of these functions vanishes identically on $\Sigma$.  When it is necessary to be more specific, we will specify that $\Sigma$ is maximal with respect to $\hat g$, $g$, or $\tilde g$ accordingly.  We will assume for the remainder of the paper that $\Sigma$ is noncompact, orientable, and geodesically complete with respect to $h$.

It will be useful for us to relate the scalar curvatures between the total space and the quotient, and additionally, to know that $(\Sigma,\hat h)$ has nonnegative scalar curvature when $\Sigma$ is a maximal hypersurface.

\begin{lem}
\label{s}
Let $\hat s$ and $s$ denote the scalar curvatures corresponding to $\hat h$ and $h$, respectively.  Then
\[
s = \hat s + 2\Delta_h\psi + 2|\nabla\psi|_h^2 + \frac{1}{2} e^{-4\psi} \omega(\nu)^2,
\]
where $\Delta_h = \emph{tr}_h \, \nabla^2$ denotes the Laplacian on $\Sigma$.  Furthermore, if $\Sigma$ is maximal, then $\hat s \geq 0$.
\end{lem}

\begin{proof}
\begin{align*}
\hat s &= |\hat k|_{\hat h}^2 - \hat \tau^2 \\
 &= |k|_h^2 + d\psi(\nu)^2 + \frac{1}{2} e^{-4\psi} \left( |\omega |_g^2 + \omega(\nu)^2 \right) - (\tau - d\psi(\nu))^2 \\
 &= s - R - 2 Ric(\nu,\nu) + \frac{1}{2} e^{-4\psi} \left( |\omega |_g^2 + \omega(\nu)^2 \right) + 2 \tau \cdot d\psi(\nu) \\
 &= s - 2\Delta_h\psi - 2|\nabla\psi|_h^2 - \frac{1}{2} e^{-4\psi} \omega(\nu)^2,
\end{align*}
where the first and third lines follow from the Gauss-Codazzi equations and the last equation is a consequence of \eqref{reve} and the fact that $\Box_g \psi = \Delta_h \psi - Dd\psi(\nu,\nu) + \tau \cdot d\psi(\nu).$  The first and second lines show that $\hat s\geq 0$ when $\hat\tau = 0$ or $\tau = 0$.

The Gauss-Codazzi equations for the conformally rescaled metric $\tilde g$ give rise to
\begin{align*}
	\tilde s &= \tilde R + 2 \tilde Ric(\tilde \nu,\tilde \nu) + |\tilde k|_{\tilde h}^2 - \tilde\tau^2 \\
	  &= 2 |d\psi |_{\tilde g}^2 + \frac{1}{2} e^{-4\psi} |\omega |_{\tilde g}^2 + 4 d\psi(\tilde \nu)^2 + e^{-4\psi} \omega(\tilde \nu)^2 + |\tilde k|_{\tilde h}^2 - \tilde\tau^2 \\
	  &= e^{-2\psi} \left( 2 |d\psi |_g^2 + \frac{1}{2} e^{-4\psi} |\omega |_g^2 + 4 d\psi(\nu)^2 + e^{-4\psi} \omega(\nu)^2 + |\tilde k|_h^2 \right) - \tilde\tau^2.
\end{align*}
Now use the conformal relation $e^{2\psi}\tilde s = s - 2\Delta_h\psi$ and the last line of the previous calculation to see that $\hat s \geq 0$ when $\tilde\tau = 0$.
\end{proof}

Now that we know our maximal hypersurfaces have nonnegative curvature, we can make use of the following result.  Even though the isometry group it refers to is $S^1$, the result still applies.  All we have to do is consider a new total space $\hat M$ in which the $\R$ fibers have been compactified to $S^1$ fibers with length $e^\psi$.

\begin{prop}[Anderson \cite{mA98}]
\label{andprop}
Let $N$ be an orientable complete Riemannian 3-manifold having nonnegative scalar curvature and admitting a free isometric $S^1$ action.  Let $V$ be the quotient space $V=N/S^1$ and let $\lambda(r)$ denote the circumference of a geodesic disc of radius $r$ centered at a fixed point in $V$.  Then there exists a constant $c$ such that
\[
\lambda(r) \leq c \cdot r.
\]
\end{prop}

Anderson also shows that square norms of tensor quantities relating to the geodesic curvature of the $S^1$ fibers and the obstruction to integrability of the horizontal distribution, in our case  $|\nabla \psi|_h^2$ and $\omega(\nu)^2$, are $L^1$ functions on $N / S^1$.  Adapting Anderson's proof, we can establish a similar result, Lemma \ref{G}, for maximal hypersurfaces of quotient vacuum spacetimes.

From here on we will assume all metric related quantities are determined by $h$ unless indicated otherwise.  Choose any point in $\Sigma$ and denote by $D(r)$ the disc of radius $r$ centered at that point.  Let $a(r)$ denote the area of $D(r)$ and let $\lambda(r)$ denote the length of the boundary $\partial D(r)$.

\begin{lem}
\label{G}
Suppose $\Sigma \subset M$ is maximal.  Define $G$ by
\[
G(r) = \int_{D(r)} \hat s + 2|\nabla \psi|^2 + \frac{1}{2} e^{-4\psi} \omega(\nu)^2 \; d\mu,
\]
where $d\mu$ is the volume form on $\Sigma$ given by $h$.
Then $G$ is bounded above by $4\pi$ and
\begin{equation}
\label{o} \int^r \left[ G'(z) \lambda(z) \right]^{1/2} \; dz = o(r),
\end{equation}
where $o(r)$ denotes a function satisfying $\lim_{r\ra\infty} o(r)/r = 0$.
\end{lem}

\begin{proof}
First of all, we claim that
\begin{equation}
\label{Deltapsi}
- \int_{D(r)} \Delta \psi \; d\mu \leq \left[ \lambda(r) G'(r) \right]^{1/2}.
\end{equation}
%
This is shown in the proof of Theorem 1 of \cite{CY75} in a more general context.  We detail the argument here.  Choose a small $\delta>0$.  Given $c>1$ we can find a smooth function $\eta$ supported on $D(r)$ which satisfies $\eta = 1$ on $D(r-\delta)$ and $|\nabla \eta| \leq \frac{c}{\delta}$.  Integrating by parts and applying the Cauchy-Schwarz inequality gives
\begin{align*}\label{psiineq}
-\int_{D(r)} \eta \cdot \Delta \psi &\leq \int_{D(r) \backslash D(r-\delta)} \langle \nabla \eta, \nabla \psi \rangle \nonumber \\
	&\leq \left[ \int_{D(r) \backslash D(r-\delta)} | \nabla \eta |^2 \right]^{1/2} \left[ \int_{D(r) \backslash D(r-\delta)} | \nabla \psi |^2 \right]^{1/2} \nonumber \\
	&\leq \left[ \left( \frac{c}{\delta} \right)^2 \left( a(r)-a(r-\delta) \right) \right]^{1/2} \Big[ G(r) - G(r-\delta) \Big]^{1/2} \nonumber \\
	&\leq  c \left[ \frac{1}{\delta} \left( a(r)-a(r-\delta) \right) \right]^{1/2} \left[ \frac{1}{\delta} \left( G(r) - G(r-\delta) \right) \right]^{1/2},
\end{align*}
where we have used the fact that $\hat s \geq 0$ from Lemma \ref{s} to obtain the factor involving $G$ in the third line.  Taking the limit $\delta \ra 0$ and then letting $c \ra 1$ establishes \eqref{Deltapsi}.

The Gauss-Bonnet theorem reads
\[
\int_{D(r)} s \; d\mu = 4\pi\chi(r) - 2\kappa(r),
\]
where $\chi(r)$ denotes the Euler characteristic of $D(r)$ and $\kappa(r)$ denotes the total geodesic curvature of the piecewise smooth curve $\partial D(r)$, including the sum of any exterior angles.  Putting the above two inequalities together with the expression for $s$ given in Lemma \ref{s}, we obtain
\begin{equation}
\label{Gineq}
G(r) \leq 4\pi\chi(r) - 2\lambda'(r) + 2\left[ \lambda(r) G'(r) \right]^{1/2},
\end{equation}
where we have used the fact that $\lambda'(r) \leq \kappa(r)$ (see \cite[p. 380]{GL83}).  The derivative $\lambda'(r)$ is to be interpreted in the sense of distributions.

If $G \leq 4\pi$ then the lemma is proved, so assume otherwise and choose $r_0$ so that $H(r) := G(r) - 4\pi > 0$ for all $r>r_0$.  Then we have
\[
\frac{1}{\lambda^{1/2}} \leq -2\frac{\lambda'}{H \lambda^{1/2}} + 2 \frac{(H')^{1/2}}{H},
\]
where we have used the fact that $\chi \leq 1$ since $D(r)$ is connected.  We now integrate each of the terms in the above inequality from $r_0$ to some value $r_1 > r_0$.  Lemma \ref{s} tells us that $\hat s \geq 0$.  So, for the leftmost term, we may use Proposition \ref{andprop} to obtain
\[
\int_{r_0}^{r_1} \lambda^{-1/2} \; dr \geq A (r_1^{1/2}-r_0^{1/2}),
\]
for some positive constant $A$.  For the middle term we have
\begin{multline*}
\int_{r_0}^{r_1} -2\frac{\lambda'}{H \lambda^{1/2}} \; dr
	\leq -2 \left( \frac{\lambda^{1/2}(r_1)}{H(r_1)} - \frac{\lambda^{1/2}(r_0)}{H(r_0)} \right) - 2 \int_{r_0}^{r_1} \frac{\lambda^{1/2}}{H^2} H' \; dr \\
	\leq 2 \frac{\lambda^{1/2}(r_0)}{H(r_0)}.
\end{multline*}
And, for the last term,
\[
\int_{r_0}^{r_1} \frac{(H')^{1/2}}{H} \; dr \leq \left( \int_{r_0}^{r_1} \frac{H'}{H^2} \; dr \right)^{1/2} \left( \int_{r_0}^{r_1} 1 \; dr \right)^{1/2} \leq \left[ H(r_0) \right]^{-1/2} r_1^{1/2}.
\]
Putting these results together gives
\[
A (r_1^{1/2}-r_0^{1/2}) \leq \frac{\lambda^{1/2}(r_0)}{H(r_0)} + 2 \left[ H(r_0) \right]^{-1/2} r_1^{1/2}.
\]
Since this is true for all $r_1 > r_0$, we must have
\[
A \leq 2 \left[ H(r_0) \right]^{-1/2}.
\]
$r_0$ can be made arbitrarily large and $A$ is independent of $r_0$.  It follows that $G$ is bounded above.

Since $G$ is bounded and non-decreasing, for any $\epsilon>0$ we can find a number $r_0$ so that $(G(r_1) - G(r_0))^{1/2} < \epsilon$ for all $r_1 > r_0$.  Therefore,
\[
\int_{r_0}^{r_1} [G' \lambda]^{1/2} \; dr \leq \left(G(r_1) - G(r_0) \right)^{1/2} \left(\frac{c}{2}r_1^2-\frac{c}{2}r_0^2 \right)^{1/2} < \epsilon \sqrt{c/2} \cdot r_1,
\]
where $c$ is the constant from Proposition \ref{andprop}.
This argument holds for all $\epsilon > 0$, so we have \eqref{o}.

Finally, going back to \eqref{Gineq}, integrating both sides from $0$ to $r$ and using \eqref{o} gives the upper bound $G \leq 4\pi$.
\end{proof}

We define the order of a complete Riemannian manifold $N$, just as in \cite{CY75}.  Let $V(r)$ denote the volume of a geodesic ball of radius $r$ centered at a fixed point $p$ in $N$.  The \emph{order of $N$} is defined to be the smallest number $O(N)$ such that
\[
\liminf_{r\ra\infty} ( V(r)/r^{O(N)} ) < \infty.
\]
With this, we can state the following Liouville type theorem.

\begin{thm}[Cheng-Yau \cite{CY75}]
\label{Liouville}
If $N$ is a complete Riemannian manifold of order $O(N) \leq 2$, then $N$ is parabolic, i.e., any subharmonic function that is bounded above on $N$ must be constant.
\end{thm}

It is now straightforward to prove Theorem \ref{topgeo}.

\textbf{Proof of Theorem \ref{topgeo}}

First consider the case where $\Sigma$ is simply connected.  Lemma \ref{s} and Proposition \ref{andprop} reveal that $(\Sigma,h)$ has order $\leq 2$, so Theorem \ref{Liouville} tells us that $(\Sigma,h)$ is parabolic.  It follows from the unifomization theory of surfaces (for example, see \cite{FK92}) that $(\Sigma,h)$ must be conformal to the Euclidean plane.
	
Next consider the case where $\Sigma$ is not simply connected.  Applying the above argument to the universal cover of $\Sigma$ shows that $(\Sigma,h)$ is conformally flat.  This leaves three possibilities for the topology of $\Sigma$: the plane, $\R^2$; cylinder, $S^1 \times \R$; or torus, $S^1 \times S^1$.  Since $\Sigma$ is neither simply connected nor compact, $\Sigma$ must be topologically equivalent to the cylinder, $S^1 \times \R$.

Define $G(r)$ as in Lemma \ref{G}.  The disc $D(r)$ will fail to be simply connected when $r$ is sufficiently large, say $r > r_0$, so $\chi(r) \leq 0$.  Therefore, integrating both sides of \eqref{Gineq} from $r_0$ to $r$, we find
\[
\int_{r_0}^r G \leq 2 \lambda(r_0) + o(r).
\]
Since $G$ is a nonnegative nondecreasing function, the above inequality can hold only if $G$ is identically zero.  This proves the second statement of the theorem.\qed

\section{Positivity of Mass}
\label{pmtsec}

A Hamiltonian framework is laid out for quotient spacetimes foliated by compact spacelike hypersurfaces in \cite{vM86}.  This is carried to the noncompact case in \cite{AV94} and \cite{mV95} and the ADM mass of an asymptotically flat two-dimensional hypersurface is formulated.  The authors argue that the usual model for asymptotic flatness, the Euclidean metric $\delta$, is not appropriate in the quotient setting discussed here.  Instead the model metric is taken to be $r^{-\beta}\delta$ for some constant $\beta$, where $\delta$ is given by $dr^2+r^2d\theta^2$ in standard polar coordinates $(r,\theta)$.  This definition is motivated by examples involving point sources and also a class of special quotient vacuum spacetimes called Einstein-Rosen waves (see \cite{ABS971}).  Note that $r^{-\beta} \delta$, for $0 \leq \beta < 2$, is the induced metric on a cone in Euclidean space with opening angle $\theta_0 = (2-\beta) \pi$.

By following the well known procedure which gives rise to ADM mass for the four-dimensional asymptotically Euclidean setting \cite{ADM62}, the authors conclude that $\beta$ is a constant multiple of the ADM mass for the quotient setting.  They arrive at a positive energy theorem: $\beta \geq 0$ with equality if and only if $(\Sigma,h,k)$ is initial data for three-dimensional Minkowski spacetime.  The authors find it necessary to require that $\beta \leq 2$ to ensure that their Hamiltonian framework remains well defined.  It is pointed out that the upper bound on $\beta$ contrasts the behavior of the ADM mass in the four-dimensional asymptotically Euclidean setting.

We define $\beta$, without imposing an asymptotic flatness assumption, as follows.

\begin{defn}
\[
\beta := \frac{1}{2\pi} \int_{\Sigma} \hat s + 2 |\nabla \psi|^2 + \frac{1}{2} e^{-4\psi} \omega(\nu)^2 \; d\mu = \frac{1}{2\pi}\int_{\Sigma} \tilde s \; d\tilde\mu,
\]
where $d\mu$ and $d\tilde\mu$ are the volume forms on $\Sigma$ corresponding to $h$ and $\tilde h$, respectively.
\end{defn}

It is pointed out in \cite{AV94} that the value of $\beta$ defined as we have done here is consistent with the definition of mass described above even though the integrand itself does not match the Hamiltonian that gives rise to the ADM mass.

We now prove Theorem \ref{pmt} which plays the role of a positive energy theorem and correlates well with the positive energy theorem derived in \cite{BCM95} where cylindrical symmetry and asymptotic flatness are assumed.

\textbf{Proof of Theorem \ref{pmt}}

\begin{proof}

Lemma \ref{G} shows that $\beta = G(\infty)/2\pi \leq 2$.  The inequality $\beta \geq 0$ follows from Lemma \ref{s} and the relation $e^{2\psi} \tilde s = s - 2\Delta_h \psi$.  The ``only if'' statement follows from the Gauss-Codazzi relations, which imply that $\hat s = |\hat k|_{\hat h}^2$ when $\hat\tau = 0$.  The definition of $\beta$ and Lemma \ref{s} then show that $\hat s = 0$.  To prove the ``if'' statement of the theorem, suppose $(\hat\Sigma,\hat h,\hat k)$ is Minkowskian initial data.  Then any Cauchy development $(\hat M, \hat g)$ of that data must be a flat spacetime, so $d\psi$ and $\omega$ must both vanish identically on $\hat M$.  Therefore $\beta=0$ on $\Sigma$.
\end{proof}

It does not seem to be as straightforward to prove Theorem \ref{pmt} for the cases $\tau = 0$ or $\tilde \tau = 0$ on $\Sigma$.  One can deduce $k = 0$ or $\tilde k = 0$, respectively, but $|\hat k|_{\hat h}^2$ will be a constant multiple of $d\psi(\nu)$.

The fact that we can formulate a positive mass theorem for maximal hypersurfaces with no condition of asymptotic flatness leads one to question what other similarities are shared by the maximality condition and the asymptotic flatness condition $h \sim r^{-\beta} \delta$.  This issue is especially pertinent if one were to carry further the Hamiltonian analysis in \cite{AV94} under the maximal hypersurface gauge assumption, as is done in \cite{mV95}.  While similarities between the two conditions may be reassuring, similarities that are too strong may suggest that the two hypotheses may, when considered together, leave little room for interesting results, or worse, lead to a contradiction.

The following proposition indicates a similarity, albeit a rather weak one, between the maximal hypersurface condition and the asymptotic flatness condition.  It states that when $\Sigma$ is maximal, the circumference $\lambda(r)$ of a geodesic disc with sufficiently large radius $r$ is no more than the circumference of a disc of equal radius in the Euclidean plane, and Euclidean-sized circumferences occur only when $h$ is exactly the Euclidean metric.  A more meaningful result would be something along the lines of $\lambda(r) \sim (2-\beta)\pi r$, but while no counterexamples to this are known to this author, neither is a proof.

\begin{prop}
\label{areaprop}
If $\Sigma$ is maximal, then
\[
\lambda(r) \leq 2\pi r + o(r).
\]
Moreover, as long as $(\Sigma,h)$ is not a flat cylinder, then one has equality if and only if $\beta=0$.
\end{prop}

\begin{proof}
The proof follows immediately from integrating $G$ as in the proof of Theorem \ref{pmt}.  This gives
\[
2\lambda \leq 4\pi r - \int_0^r G + o(r).
\]
\end{proof}

\section{Foliations and the Lapse Function}
\label{lapsesec}

In this section we prove Theorem \ref{maxgaugethm}.  We assume that $M$ is diffeomorphic to a product $\Sigma \times \R$ and each slice $\Sigma_t = \Sigma \times \{ t \}$ is a noncompact maximal hypersurface, geodesically complete with respect to the metric $h$ induced by $g$.  The \emph{lapse} (or \emph{lapse function}) on a hypersurface $\Sigma_t$ is the function $u:\Sigma_t \ra \R$ given by
\[
u = \left(-|dt|_g^2 \right)^{-1/2},
\]
where $t$ is the function on $M=\Sigma \times \R$ defined by $(p,x) \ra x$.

If each hypersurface is assumed to be maximal with respect to $g$, then one finds (see, for example, \cite{KM96}) that $u$ satisfies
\begin{equation}
\label{lapseeqn}
\Delta_h u - \left( |k|_h^2 + Ric(\nu,\nu) \right) u = 0.
\end{equation}
This equation is known as the \emph{lapse equation}.  Analogously, we have lapse equations for the cases $\hat \tau = 0$ and $\tilde \tau = 0$, respectively:
\begin{gather*}
\Delta_{\hat h} u - |\hat k|_{\hat h}^2 u = 0; \\
\Delta_{\tilde h} u - \left( |\tilde k|_{\tilde h}^2 + \tilde Ric(\tilde \nu,\tilde \nu) \right) u = 0,
\end{gather*}
where, in the $\hat\,$ case, we deal with the pullbacks via $\pi_\Sigma$ of functions on $\Sigma$.

The following Lemma will help us in dealing with solutions of the lapse equation.

\begin{lem}
\label{lapselem}
Let $N$ be a complete Riemannian manifold of order $O(N)\leq 2$.  Consider functions $w : N \ra [0,\infty)$ and $v : N \ra (c_1,c_2)$, where $c_1$ and $c_2$ are positive constants.  Then the only positive bounded solutions $u : N \ra (0,\infty)$ to the equation
\[
\frac{\Delta u}{u} = \frac{\Delta v}{v} + w
\]
are constant multiples of $v$.  Note that the existence of a positive bounded solution implies that $w$ is the zero function.\footnote{If it happens that $w$ is zero from the outset, then the boundedness requirement can be removed from $u$.  The same proof goes through with $\gamma = \ln (u/v)$.}
\end{lem}

\begin{proof}
Let $c_0$ be an upper bound for $u/v$ and put $\gamma = \ln\left( c_0 - u/v \right)$.  Then we have
\begin{align}
\label{legamma}
\Delta \gamma + |\nabla \gamma|^2 &= -2 \langle \nabla \gamma,\nabla \ln v \rangle - (c_0 e^{-\gamma} - 1) \cdot w \leq -2 \langle \nabla \gamma,\nabla \ln v \rangle
\end{align}
Let $n$ be a positive integer.  Multiply both sides of the above inequality by $(\ln v)^n$ and integrate over a geodesic ball $B=B(r)$ of radius $r$.  Integration by parts gives
\begin{multline*}
n \int_B (\ln v)^{n-1} \langle \nabla \gamma, \nabla \ln v \rangle \geq
	\int_{\partial B} (\ln v)^n \partial_r \gamma + \int_B (\ln v)^n |\nabla \gamma|^2 \\
	+ 2\int_B (\ln v)^n \langle \nabla \gamma, \nabla \ln v \rangle.
\end{multline*}
Here $\partial_r$ denotes differentiation with respect to the outward unit normal vector along $\partial D$, and these integrals are in terms of the volume form given by the metric on $N$.  The relation can be iterated to obtain
\begin{multline*}
-2 \int_B \langle \nabla \gamma, \nabla \ln v \rangle \leq - \sum_{n=1}^m \frac{2^n}{n!} \left[ \int_{\partial B} (\ln v)^n \partial_r \gamma + \int_B (\ln v)^n |\nabla \gamma|^2 \right] \\
- \frac{2^{m+1}}{m!} \int_B (\ln v)^m \langle \nabla \gamma,\nabla \ln v \rangle.
\end{multline*}
Letting $m \ra \infty$, this becomes
\[
-2 \int_B \langle \nabla \gamma, \nabla \ln v \rangle \leq \int_{\partial B} (1-v^2) \partial_r \gamma + \int_B (1-v^2) |\nabla \gamma|^2.
\]
Now put this together with equation \eqref{legamma} to get
\begin{align*}
\int_B |\nabla \gamma|^2 &\leq \int_{\partial B} (1-v^2) \partial_r \gamma + \int_B (1-v^2)|\nabla \gamma|^2 - \int_B \Delta \gamma \\
	&= -\int_{\partial B} v^2 \partial_r \gamma + \int_B (1-v^2)|\nabla \gamma|^2,
\end{align*}
and this gives
\[
\int_B v^2 |\nabla \gamma|^2 \leq - \int_{\partial B} v^2 \partial_\nu \gamma \leq \left( \int_{\partial B} v^4 \right)^{1/2} \left( \int_{\partial B} |\nabla \gamma|^2 \right)^{1/2}.
\]
Put $\Gamma(r) = \int_{B(r)} |\nabla \gamma|^2$ and let $V(r)$ denote the volume of $B(r)$.  Then we have
\[
c_1^2 \, \Gamma \leq c_2^2 \left( V' \Gamma' \right)^{1/2}.
\]

Next we show that a contradiction arises if $\Gamma$ is not identically zero.  We follow the same arguments in the proof of Theorem \ref{Liouville} in \cite{CY75}.  If $\Gamma \neq 0$, then
\[
\left( \frac{c_2}{c_1}\right)^4 \frac{\Gamma'}{\Gamma^2} \geq \frac{1}{V'}.
\]
Integrating both sides gives
\begin{multline*}
\left( \frac{c_2}{c_1}\right)^4 \left( \frac{1}{\Gamma(r_0)}-\frac{1}{\Gamma(r)} \right) \geq \int_{r_0}^{r} \frac{1}{V'(z)} \,dz \geq \left( \int_{r_0}^r dz \right)^2 \bigg/ \int_{r_0}^{r} V'(z) \,dz \\
\geq \frac{(r-r_0)^2}{V(r)-V(r_0)}.
\end{multline*}
Since $(\Sigma,h)$ has order $\leq 2$, there exists a constant $c > 0$ and a sequence $0 < r_0 < r_1 < r_2 < \cdots$ with $V(r_i) < c r_i^2$.  Then
\[
\left( \frac{c_2}{c_1}\right)^4 \frac{1}{\Gamma(r_0)} \geq \sum_{i=0}^\infty \frac{(r_{i+1}-r_i)^2}{V(r_{i+1})-V(r_i)} > \frac{1}{c}\sum_{i=0}^{\infty} \left( 1 - \frac{r_i}{r_{i+1}} \right)^2.
\]
By passing to a subsequence, we may arrange that $r_i/r_{i+1} < 1/2$, which gives a contradiction.  Therefore $\Gamma$ is identically zero.  This completes the proof.\end{proof}

Lemma \ref{lapselem} is the tool we need to prove Theorem \ref{maxgaugethm}.
 
\noindent \textbf{Proof of Theorem \ref{maxgaugethm}}

We begin by considering the case for which $\Sigma_t$ is maximal with respect to $g$ for all $t$.  Afterwards we will show that the cases of maximality with respect to $\hat g$ and $\tilde g$ essentially boil down to this case.  Put $f=e^\psi$.  We use the fact that $\Box_g \psi = \Delta_h \psi - Dd\psi(\nu,\nu) + \tau \cdot d\psi(\nu)$ together with \eqref{reve} to see that
\begin{align*}
Ric(\nu,\nu) &= D d \psi(\nu,\nu) + d\psi(\nu)^2 \\
	&= \Delta_h \psi - \Box_g \psi + d\psi(\nu)^2 \\
	&= \Delta_h \psi + |\nabla\psi|_h^2 \\
	&= \Delta_h f / f.
\end{align*}
Therefore, for any particular value of $t$, the lapse equation \eqref{lapseeqn} on $\Sigma_t$ takes the form
\[
\Delta_h u/u = \Delta_h f / f + |k|_h^2.
\]
We apply Lemma \ref{lapselem} with $v=f$ and $w = |k|_h^2$ to conclude that $k$ vanishes identically on $\Sigma_t$ and that $u$ is a constant multiple of $f$ on $\Sigma_t$.  Through an appropriate reparameterizion of $t$, we may arrange that $u$ is in fact \emph{equal} to $f$ on each $\Sigma_t$, so we shall do so without loss of generality.

%
It follows from Theorem \ref{topgeo} that there exist coordinates $x,y$ globally defined on $\Sigma_0$ and a function $\alpha:\Sigma_0 \ra \R$ so that $h = e^{2\alpha} \left( dx^2 + dy^2 \right)$, where $dx^2 + dy^2$ represents either the Euclidean metric on $\Sigma_0 = \R^2$ or else a flat metric on $\Sigma_0 = S^1 \times \R$.  (One may assume $\alpha \equiv 0$ in the latter case.)  Let $V$ be a vector field on $M$ which is orthogonal to each hypersurface $\Sigma_{t}$ and satisfies $dt(V)=1$.  The functions $\alpha,x,y$ can pulled back from $\Sigma_0$ to any other hypersurface $\Sigma_{t}$ via the flow generated by the vector field $V$ (or $-V$).  Then $\{t,x,y\}$ is a coordinate system globally defined on $M = \Sigma \times \R$.  In these coordinates, the metric $g$ takes the form
\[
g = -e^{2\psi(t,x,y)} dt^2 + e^{2\alpha(x,y)} \left( dx^2 + dy^2 \right).
\]
The function $\alpha$ is independent of $t$ because, as shown above, $\Sigma_t$ is totally geodesic for all $t$.

From the Gauss-Codazzi relation $\hat s = |\hat k|_{\hat h}^2 - \hat \tau^2$ and the relations \eqref{ktokhat} and \eqref{tautotauhat},  we see that $\hat s = 0$.  Therefore, from Lemma \ref{s}, we have
\begin{equation}
\label{psitoalpha}
e^{-\psi} \Delta_h e^\psi = \Delta_h \psi + |\nabla \psi|_h^2 = s/2 = -\Delta_h \alpha.
\end{equation}
This shows that $e^{-\psi} \Delta_h e^\psi$ is independent of $t$.  Thus, for any $t$, setting $u=e^{\psi(0,x,y)}$, $v=e^{\psi(t,x,y)}$, and $w=0$ in Lemma \ref{lapselem} shows that $u$ is a constant multiple of $v$.  Consequently, there is a smooth function $\rho:\R \ra \R$ such that $\psi(t,x,y) = \psi(0,x,y) + \rho(t)$ for all $t$.


From \eqref{reve}, we have
\[
Ric(\nu,\nu) - d\psi(\nu)^2 = Dd\psi(\nu,\nu) = \nu(d\psi(\nu)) - d\psi(D_\nu \nu).
\]
It is straightforward to verify that $d\psi(D_\nu \nu) = |\nabla\psi|_h^2$.  Combining this with the relation $Ric(\nu,\nu) = \Delta_h \psi + |\nabla\psi|_h^2$ established above, we arrive at
\[
\Delta_h e^{2\psi} = 2e^{2\psi} \left[ \Delta_h \psi + 2|\nabla\psi|_h^2 \right] = 2e^{2\psi} \left[ d\psi(\nu)^2 + \nu(d\psi(\nu)) \right] = 2\frac{d^2\rho}{dt^2}.
\]
This shows that $\Delta_h e^{2\psi}$ is constant in $x$ and $y$.  As shown in the proof of Theorem \ref{topgeo}, we know that $\Sigma_0$ has order $O(N)\leq 2$.  Therefore, we deduce from Theorem \ref{Liouville} that $\psi$ is in fact constant on $\Sigma_0$.  This implies that $\rho$ is a linear function of $t$ and that
\[
\psi = bt+ c,
\]
where $b$ and $c$ are constants on $M$.

Returning to equation \eqref{psitoalpha}, we see that $\alpha$ is harmonic and must therefore be a linear function of $x$ and $y$.  It follows that the nontriviality of $\alpha$ can be washed away through a judicious choice of the coordinate functions $x$ and $y$, so we will assume, without loss of generality, that $\alpha$ is identically zero.  Thus we have
\[
\hat g =  e^{2(bt+c)} \left( -dt^2 + dz^2 \right) + dx^2 + dy^2.
\]
If $b=0$, then make the coordinate change $\tilde t = e^c t$ and $\tilde z = e^c z$, otherwise set $\tilde t = \frac{1}{b} e^{bt+c} \cosh (bz)$ and $\tilde z = \frac{1}{b} e^{bt+c} \sinh (bz)$.  In either case, $\hat g$ takes the Minkowskian form
\[
\hat g = - d\tilde t^2 + d\tilde z^2 + dx^2 + dy^2.
\]
This proves the theorem when each $\Sigma_t$ is maximal with respect to $g$.

Next consider case where $\Sigma_t$ is maximal with respect to $\hat g$ for all $t$.  The lapse equation reads
\[
\Delta_{\hat h} u = |\hat k|_{\hat h}^2 u.
\]
Since, for each $t$, $\hat \Sigma_t$ is assumed to be the pullback $\pi_M^{-1}(\Sigma_t)$, the function $u$ must be invariant with respect to the Killing action on $\hat \Sigma_t$.  Consider the three-manifold obtained by compactifying the $\R$ fibers of the bundle $\hat\Sigma_t \ra \Sigma_t$ so that they become $S^1$ fibers (of length $e^\psi$).  There will be nothing lost in assuming $(\hat\Sigma,\hat h)$ has this form.  Then an upper bound on $\psi$ guarantees that $(\hat \Sigma_t,\hat h)$ has order $O(\hat\Sigma_t)\leq 2$.  We may apply Theorem \ref{Liouville} (or Lemma \ref{lapselem} with $v \equiv 1$) to conclude that $|\hat k|^2$ is identically zero and $u$ is constant on $\hat \Sigma_t$.  Equation \eqref{ktokhat} shows that $k$ vanishes identically on $\Sigma_t$, and this is true for each $t$.  Thus we have found ourselves in the case ($\tau  \equiv 0$) for which we have already proven the theorem.

Finally consider the case where $\Sigma_t$ is maximal with respect to $\tilde g$ for all $t$.  Equation \eqref{cric} shows that the lapse equation reads
\[
\Delta_{\tilde h} u = \left( |\tilde k|_{\tilde h}^2 + 2 d\psi(\tilde \nu)^2 \right) u.
\]
Theorem \ref{Liouville} (or Lemma \ref{lapselem} with $v \equiv 1$) shows that $u$ is constant on $\Sigma_t$, and both $|\tilde k|_{\tilde h}^2$ and $d\psi(\tilde \nu)^2$ vanish identically on $\Sigma_t$.  From equation \eqref{ktoktilde} we see that $k$ is identically zero on $\Sigma_t$.  Again we find ourselves in the case ($\tau  \equiv 0$) for which we have already proven the theorem.
\qed

\bibliographystyle{plain}
\bibliography{maxhyp}

\end{document}